\documentclass{eptcs}
\providecommand{\event}{TERMGRAPH 2011} 

\usepackage{amsthm}

\theoremstyle{plain}
\newtheorem{theorem}{Theorem} 

\newtheorem{proposition}[theorem]{Proposition}

\theoremstyle{definition}
\newtheorem{definition}[theorem]{Definition}
\newtheorem{example}[theorem]{Example}

\def\squareforqed{\hbox{\rlap{$\sqcap$}$\sqcup$}}
\def\qed{\ifmmode\squareforqed\else{\unskip\nobreak\hfil
\penalty50\hskip1em\null\nobreak\hfil\squareforqed
\parfillskip=0pt\finalhyphendemerits=0\endgraf}\fi}

\usepackage{breakurl}             
\usepackage[all]{xy}
\usepackage{latexsym}
\usepackage{graphicx}

\newcommand{\varnodes}[1]{var(#1)}
\newcommand{\nat}{{\rm I\!l\!\!N}}

\title{Term Graph Rewriting and Parallel Term Rewriting}

\author{Andrea Corradini
\institute{Universit\`{a} di Pisa\\ Dipartimento di Informatica}
\email{andrea@di.unipi.it}
\and Frank Drewes
\institute{Department of Computing Science\\Ume\r{a} University}
\email{drewes@cs.unu.se}}
\def\titlerunning{Term Graph Rewriting and Parallel Term Rewriting}
\def\authorrunning{A. Corradini and F. Drewes}

\begin{document}
\maketitle
 
\begin{abstract}
The relationship between Term Graph Rewriting and Term Rewriting is
well understood: a single term graph reduction may correspond to
several term reductions, due to sharing.  It is also known that if
term graphs are allowed to contain cycles, then one term graph
reduction may correspond to infinitely many term reductions. We stress
that this fact can be interpreted in two ways. According to the
\emph{sequential interpretation}, a term graph reduction corresponds
to an infinite sequence of term reductions, as formalized by Kennaway
et.\ al.\ using strongly converging derivations over the complete
metric space of infinite terms. Instead according to the
\emph{parallel interpretation} a term graph reduction corresponds to
the parallel reduction of an infinite set of redexes in a rational
term. We formalize the latter notion by exploiting the complete
partial order of infinite and possibly partial terms, and we stress that
this interpretation allows to explain the result of reducing circular
redexes in several approaches to term graph rewriting.

\end{abstract}

\section{Introduction}

The theory of {\em Term Graph Rewriting\/} (TGR) studies the issue 
of representing finite terms with directed, acyclic graphs, and of
modeling term rewriting via graph rewriting.
This field has a long history in the realm of theoretical
computer science, its origin dating back to the seventies of the last
century, when dags were proposed in~\cite{DBLP:conf/icalp/PaciniMT74}
as an efficient implementation of recursive program schemes.
Among the many contributions to the foundations of this field, we mention~\cite{Sta:CGLE,BEGKPS:TGR,Ken:GRSC,HP:ITRJ,CR:HRJR,DBLP:journals/fuin/AriolaK96,CG:A2CP}.  
The various approaches may differ for the way term graphs are
represented or for the precise definition of the graph rewriting
mechanism, but they all present equivalent results for what
concerns the speed-up of term rewriting due to the explicit sharing.

In fact, the  main advantage of using graphs is that when applying a rewrite
rule, the subterm matched by a variable $x$ of the left-hand side does
not need to be copied if $x$ appears more than once in the right-hand
side, because the sharing of subterms can be represented
explicitly. Therefore the
rewriting process is speeded up, because the rewriting steps do not
have to be repeated for each copy of a subterm.

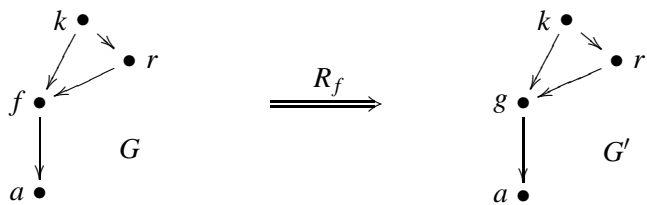
\begin{figure}[b]
\centering
$
\xymatrix@=1ex{
& \save []-<.3cm, .0cm>*\txt{$k$} \restore \bullet \ar[dr] \ar[ldd] & \\
& & \save []+<.3cm, .0cm>*\txt{$r$} \restore \bullet \ar[dll]\\
\save []-<.3cm, .0cm>*\txt{$f$} \restore\bullet \ar[dd]& & \\
& &  G\\
\save []-<.3cm, .0cm>*\txt{$a$} \restore \bullet & &
}
$
\hspace{1cm}
$
\xymatrix@=2ex{
& \\
& \\
\mbox{ }\ar@{=>}[rrr]^{\mbox{ $R_f$}} & & & \mbox{ } \\
& \\
& \\
}
$
\hspace{1cm}
$
\xymatrix@=1ex{
& \save []-<.3cm, .0cm>*\txt{$k$} \restore \bullet \ar[dr] \ar[ldd] & \\
& & \save []+<.3cm, .0cm>*\txt{$r$} \restore \bullet \ar[dll]\\
\save []-<.3cm, .0cm>*\txt{$g$} \restore\bullet \ar[dd]& & \\
& &  G'\\
\save []-<.3cm, .0cm>*\txt{$a$} \restore \bullet & &
}
$
\caption[ ]{An example of term graph rewriting}
\label{fi:term-graph-1}
\end{figure}

For example, suppose that rule $s(x) \rightarrow k(x,r(x))$ is applied
to term $s(f(a))$, obtaining term $t = k(f(a), r(f(a)))$. Now 
rule $R_f : f(x) \rightarrow g(x)$ can be applied twice to term
$t$, yielding in two steps term $t' =
k(g(a), r(g(a))$. If instead $t$ is represented as a graph where the
two identical subterms are shared (as in graph $G$ of Figure 
\ref{fi:term-graph-1}), then a single application of the rule is 
sufficient to
reduce it to graph $G'$ of Figure~\ref{fi:term-graph-1},  which 
clearly
represents term $t'$. Thus a single graph reduction may
correspond to $n$ term reductions, where $n$ is the ``degree of
sharing'' of the reduced subterm.

Often this fact is spelled out by observing that a single TGR
reduction corresponds to \emph{a sequence} of $n$ term reductions~\cite{BEGKPS:TGR,HP:ITRJ}, but one may equivalently think 
that the $n$ term reductions are performed \emph{in parallel}, in a
single step. The idea of reducing families of redexes in
parallel (\emph{family reductions}) was proposed already in~\cite{DBLP:journals/jcss/Vuillemin74,DBLP:journals/jacm/BerryL79} as
an alternative to the use of dags for representing the sharing of subterms.

Some authors considered the extension of term
graph rewriting to the cyclic case, allowing (finite, directed)
cyclic graphs as well. Actually, already in~\cite{BEGKPS:TGR} the
definition did not forbid cycles, and the relationship with 
term rewriting in this case was analyzed in depth
in~\cite{KKSV:AGRS}, as discussed below. It was observed that
allowing term graphs with cycles one could represent certain
structures that arise when dealing with recursive definitions (as for
the implementation of the fixed point operator {\sf Y} proposed
in~\cite{DBLP:journals/spe/Turner79}). Interestingly, even using
acyclic rules, cyclic graphs can be produced by rewriting in
presence of suitable sharing strategies~\cite{FRW:CPCR,FW:RCTG}.
Cyclic term graph rewriting was also defined abstractly in a
categorical setting in~\cite{CG:RCSE}, showing the equivalence with an
operational definition; and it was discussed in the framework of
Equational Term Graph Rewriting in~\cite{DBLP:journals/fuin/AriolaK96}.

A renewed interest in term graph rewriting with cycles is witnessed by
some recent publications. In~\cite{DBLP:journals/mscs/BaldanBCK07,DBLP:journals/entcs/BaldanBCK08}
the authors propose an extension of the $\rho$-calculus where the
sharing of subterms can be modeled explicitly and cyclic definitions
are allowed. In~\cite{DBLP:journals/entcs/DuvalEP07} cyclic term
graphs are used to represent data structures, 
and rewriting models the transformation of such structures with both
local and global redirection of pointers. In~\cite{DBLP:journals/mscs/DoughertyLL06}
they are used for the definition of the type system of an object
oriented language. 

Considering the relationship with term rewriting, the first 
consequence of the extension to cyclic term graphs is
that infinite terms (or, more precisely, \emph{rational} terms,
i.e., infinite terms with a finite number of distinct subterms) can
be represented as well. The second effect is that
a single graph reduction may now correspond to some infinite
term rewriting. Consider for example rule $R_f$ above: by applying it
to  graph $H$ of Figure \ref{fi:term-graph-2} one obtains the graph 
$H'$. Clearly, $H$
represents the infinite term $f^\omega = f(f(f(\ldots)))$, while $H'$
represents term $g^\omega$. As for the finite case, there are  
two possible ways of
interpreting the rewriting of term $f^\omega$ to term $g^\omega$ via
infinitely many applications of rule $R_f$:

\begin{figure}
\centering
$
\SelectTips{cm}{}
\xymatrix@=1ex{
f \bullet \ar `r[d] `[] `[] []& \mbox{ } \\
\mbox{ } &  \\
H\\}
$
\hspace{1cm}
$
\xymatrix@=2ex{
& \\
\mbox{ }\ar@{=>}[rrr]^{\mbox{ $R_f$}} & & & \mbox{ } \\
& \\
}
$
\hspace{1cm}
$
\SelectTips{cm}{}
\xymatrix@=1ex{
g \bullet \ar `r[d] `[] `[] []& \mbox{ } \\
\mbox{ } & \\
H'\\}
$
\caption[ ]{An example of cyclic term graph rewriting}
\label{fi:term-graph-2}
\end{figure}
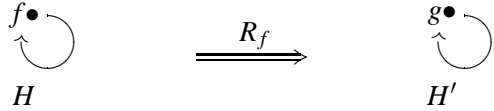

\begin{description}
\item[{[sequential interpretation]}]
$g^\omega$ is the limit of an infinite sequence of applications of
$R_f$, i.e., $f^\omega \rightarrow_{R_f} g(f^\omega) \rightarrow_{R_f}
g(g(f^\omega)) \rightarrow_{R_f}\ldots \leadsto_\omega g^\omega$.

\item[{[parallel interpretation]}]
$g^\omega$ is the result of the simultaneous application of $R_f$
to an infinite number of redexes in $f^\omega$: in a single step all
the occurrences of $f$ in $f^\omega$ are replaced by $g$.
\end{description}

\noindent
The first interpretation has been thoroughly formalized in~\cite{DK:RRRR,DKP:INFP,FRW:CPCR,FW:RCTG,KKSV:AGRS}, where the authors
elaborated a theory of (rational) \emph{transfinite term rewriting},
showing that finite, cyclic graph rewriting is an adequate
implementation for it. In essence, a finite graph derivation sequence
has the ``same effect'' of a \emph{strongly convergent} infinite
term rewriting sequence. As far as the notion of convergence is
concerned, the well-known topological structure of (possibly infinite)
terms is used, which, equipped with a suitable notion of distance,
form a complete ultra-metric space~\cite{AN:MSIT}.

In this paper we provide instead a formalization of the \emph{parallel
interpretation} above. To this aim we exploit the theory of \emph{infinite
parallel term rewriting} that has been proposed by the first author
in~\cite{Cor:TRC} 
by exploiting the complete partially ordered structure of
$CT_\Sigma$, the collection of possibly infinite, possibly partial
terms over signature $\Sigma$.

Interestingly, unlike the acyclic case for which the sequential and
the parallel interpretations of term graph rewriting are completely
equivalent, in the cyclic case there are cases
where the two interpretations lead to different results. This
happens when {\em collapsing\/} rules are considered, 
i.e.,
rules having a variable as right-hand side.
The canonical collapsing rule is the rule for identity, $R_I: I(x)
\rightarrow x$, and the pathological case (considered already by many
authors) is the application of $R_I$ to $I^\omega$. Using the sequential
interpretation, we have that $I^\omega \rightarrow_{R_I}
I^\omega \rightarrow_{R_I} \ldots$, and thus the limit of the
sequence is $I^\omega$ itself. Instead according to the parallel
interpretation all the occurrences of $I$ in $I^\omega$ are deleted
in a single step, and thus we should obtain as result a term 
that does not contain any function symbol: we will show indeed that we get
the completely undefined term $\bot$, the bottom element of the
complete partial ordering $CT_\Sigma$. 

Both the sequential and the parallel interpretations turn out to be  
meaningful from the point of view of cyclic term graph rewriting, 
because depending on the chosen rewriting approach one can get 
different results when applying the collapsing rule $R_I$ to the
\emph{circular-$I$}, i.e., to the graph having a single node labelled by
$I$ and a loop, which clearly unravels to $I^\omega$.
In fact, if one uses the operational definition of term graph
rewriting proposed in~\cite{BEGKPS:TGR} (as done  
in~\cite{KKSV:AGRS,FRW:CPCR,FW:RCTG}) then 
the circular-$I$  reduces via $R_I$ to itself, and this is 
consistent with the sequential interpretation.

Instead for several other definitions of term graph rewriting
(including the double-pushout~\cite{HP:ITRJ,CR:HRJR}, the
single-pushout~\cite{Ken:GRSC}, the
equational~\cite{DBLP:journals/fuin/AriolaK96}, an the
categorical~\cite{CG:RCSE} approaches) the circular-$I$ rewrites via
$R_I$ to a graph consisting of a single, unlabeled node, which can be
regarded as the graphical representation of the undefined term $\bot$,
and is therefore consistent with the parallel interpretation.

The paper is organized as follows. In Section 
\ref{se:term rewriting} we summarize the basic definitions about
infinite  terms~\cite{ADJ:IASCA}, orthogonal term rewriting~\cite{HL:CORS1},
and  parallel term rewriting~\cite{Cor:TRC}. 
In Section \ref{se:Term Graphs and Rational Terms} we introduce 
(possibly cyclic)
term graphs, and we make precise their relationship with (sets of)
rational terms, via the unraveling function. 
Algebraic term graph rewriting is the topic of Section
\ref{se:Algebraic term  graph rewriting}, where we recall the basics
of the  double-pushout approach~\cite{Ehr:TIAA}, apply it to the category
of term  graphs, and provide an encoding of term rewrite rules as
graph  rules. The main results of the paper are in Section
\ref{se:adequacy}.  We first prove that a single reduction of a graph
induces on the unraveled term a possibly infinite parallel
reduction. Next this fact is used as a main lemma in the proof that
the unraveling function is an adequate mapping from any orthogonal
TGRS system to the orthogonal TRS obtained by unraveling its rules,
provided that only rational terms and rational parallel  reduction
sequences  are considered.    
Finally in Section \ref{se:conclusion} we summarize our contribution.


\section{Infinite terms and parallel term rewriting}
\label{se:term rewriting}

In this background section we first introduce the algebra $CT_\Sigma$ 
of  
possibly partial, possibly 
infinite terms. The collection of such terms forms a 
complete partial ordering which has been studied in depth in~\cite{ADJ:IASCA}. Next we introduce the basic definitions related to 
(orthogonal) term rewriting, which apply to infinite terms as well. 
Finally, we introduce the definition of 
infinite parallel rewriting, which exploits the CPO structure of terms 
according to~\cite{Cor:TRC}.

\subsection{Infinite Terms}
\label{ss:Infinite Terms}

Most of the following definitions are borrowed from~\cite{ADJ:IASCA}.

Let $\omega^*$ be the set of all finite strings of
positive natural numbers. Elements of $\omega^*$ are called {\bf 
occurrences}.
The empty string is denoted by $\lambda$, and $u \le w$ indicates that 
$u$ is a prefix of $w$. 
Occurrences
$u$, $w$ are called {\bf disjoint} (written $u | w$) if neither $u \le 
w$
nor $w \le u$. 

Let $\Sigma$ be a (one-sorted) signature, i.e., a ranked alphabet of 
operator
symbols $\Sigma = \cup_{n\in\nat}\Sigma_n$ and let $X$ be a set of
variables. 
A {\bf term} over ($\Sigma$,
$X$) is a partial function $t: \omega^* \rightarrow\Sigma\cup X$,
such that the
domain of definition of $t$, $\mathcal{O}(t)$, satisfies the following 
(where $w \in\omega^*$ and all $i \in\omega$):
\begin{itemize}
\item $wi \in\mathcal{O}(t) \Rightarrow w \in\mathcal{O}(t)$

\item $wi \in\mathcal{O}(t) \Rightarrow t(w) \in\Sigma_n$
for some $n \ge i$.
\end{itemize}

$\mathcal{O}(t)$ is called the {\bf set of occurrences} of 
$t$. 
We denote by $\bot$ (called {\bf bottom}) the empty term, i.e. the 
only term such that $\mathcal{O(\bot)} = \emptyset.$

Given an occurrence $w \in\omega^*$ and a term $t$, the {\bf subterm}  
of $t$ at (occurrence) $w$
is the term
$t / w$ defined as $t / w(u) = t(wu)$ for all $u \in\omega^*$.
A term $t$ is {\bf finite\/} if $\mathcal{O}(t)$ is finite; 
it  is {\bf
total\/} if $t(w) \in\Sigma_n\Rightarrow wi \in
\mathcal{O}(t)$ for all $0 < i \le n$;
it is {\bf linear} if no variable occurs more than once in it;
and it is {\bf rational} if it has a finite number of different
subterms.
Given terms $t, s$ and an occurrence $w \in\omega^*$, the {\bf
replacement} of $s$ in $t$ at (occurrence) $w$,
denoted $t[w \leftarrow
s]$, is the term defined as $t[w \leftarrow s](u) = t(u)$ if $w
\not\le u$ or $t / w = \bot$, and $t[w \leftarrow s](wu) = s(u)$
otherwise.  

The set of terms over ($\Sigma$, $X$) is denoted by
$CT_\Sigma(X)$ ($CT_\Sigma$ stays for
$CT_\Sigma(\emptyset)$).
Throughout the paper we will often use (for finite terms) the
equivalent and more usual representation of terms as operators applied
to other terms. Partial terms are made total in this representation by
exploiting the empty term $\bot$. Thus, for example, if $x \in X$, $t = f(\bot,g(x))$
is the term such that $\mathcal{O}(t) = \{\lambda, 2, 2\cdot 1\}$, 
$t(\lambda)
= f \in\Sigma_2$, $t(2) = g \in\Sigma_1$, and $t(2\cdot 1)
= x \in X$.


It is well known that $CT_\Sigma(X)$ forms a {\bf complete partial 
order} 
with respect to the ``approximation'' relation. We say that $t$ {\bf 
approximates} $t'$ (written $t \leq t'$) iff $t$ is less defined than     
$t'$ as partial function.
The least element of $CT_\Sigma(X)$ with respect to\ $\le$ is clearly 
$\bot$. An {\bf $\omega$-chain}
$\{t_i\}_{i<\omega}$ is an infinite sequence of terms
$t_0\le t_1 \le \ldots$. Every $\omega$-chain
$\{t_i\}_{i<\omega}$ in $CT_\Sigma(X)$ has a {\bf least
upper bound} (lub) $\bigcup_{i<\omega}\{t_i\}$ characterized as
follows:
\[t = \bigcup_{i<\omega}\{t_i\} \quad \Leftrightarrow \quad \forall w
\in \omega^*\,.\, \exists i < \omega \,.\, \forall j
\ge i\,.\, t_j(w) = t(w)\]

Moreover, every pair of terms has a greatest lower bound.  All this 
amounts 
to say that
$CT_\Sigma(X)$ is an {\bf $\omega$-complete lower semilattice}.

\subsection{Term Rewriting}
\label{ss:Term Rewriting}

We recall here the basic definitions of (orthogonal) term rewriting~\cite{HL:CORS1},  
which apply to infinite terms as well.

Let $X$ and $Y$ be
two sets of variables. A {\bf substitution} (from $X$ to $Y$) is
a function $\sigma: X \rightarrow CT_\Sigma(Y)$ (used in postfix 
notation). 
Such a substitution $\sigma$ can be extended in a unique way to
a continuous (i.e., monotonic and lub-preserving) function $\sigma : 
CT_\Sigma(X) \rightarrow CT_\Sigma(Y)$, which extends $\sigma$ as 
follows

\begin{itemize}

\item $\bot\sigma = \bot$,

\item $f(t_1, ..., t_n)\sigma = f(t_1\sigma,
..., t_n\sigma)$,

\item $\left(\bigcup_{i<\omega} \{t_i\}\right)\sigma = 
\bigcup_{i<\omega} 
\{t_i\sigma\}$.
\end{itemize}


\noindent
A {\bf rewrite rule} $R = (l, r)$ is a pair of total terms of 
$CT_\Sigma(X)$,
where $var(r) \subseteq var(l)$, $l$ is finite and it is not a 
variable.\footnote{The restriction to left-finite rules can be 
motivated 
intuitively, as in~\cite{KKSV:AGRS}, by the requirement of checking in 
finite 
time the applicability of a rule to a term. On a more technical 
ground, in~\cite{Cor:TRC} it is shown that point 3 of Theorem \ref{th:WD-ParRdxApp}
below does not hold for left-infinite rules.}
Terms $l$ and
$r$ are called the {\bf left-} and the {\bf right-hand side} of $R$, 
respectively. 
A rule is called {\bf left-linear} if  $l$ is linear, and it is {\bf 
collapsing} if $r$ is a variable. 
A {\bf term rewriting 
system} (shortly {\bf TRS})
$\mathcal{R}$ is a finite set of rewrite rules, $\mathcal{R} = \{R_i\}_{i 
\in 
I}$.

Given a term rewriting system
$\mathcal{R}$, a {\bf redex} (for REDucible EXpression)  $\Delta$ of a 
term 
$t$  is a pair 
$\Delta = (w, R)$ where $R: l \rightarrow r \in \mathcal{R}$ is a rule,  and
$w$ is an occurrence of $t$, such that there exists a substitution 
$\sigma$ which {\bf realizes} $\Delta$, i.e., 
such that $t/w = l\sigma$. In this case we say that $t$ {\bf reduces} 
(via $\Delta$) to 
the 
term $t' = t[w 
\leftarrow r\sigma]$, and we write  $t \rightarrow_\Delta 
t'$ or simply $t \rightarrow t'$. 
A {\bf reduction sequence} $t_1 \rightarrow t_2 \rightarrow \ldots
t_n$ is a finite sequence of reductions. 

A TRS $\mathcal{R}$ is {\bf orthogonal} (shortly, it is an {\bf OTRS}) 
if all the rules in $\mathcal{R}$ are 
left-linear and it is {\bf non-overlapping}; that is, the 
left-hand 
side of each rule does not unify with a non-variable subterm 
of 
any other rule in $\mathcal{R}$, or with a proper, non-variable subterm
of itself. In this paper we will be concerned  
with orthogonal TRS's only, because the confluence of such 
systems is a key property needed in the next section to 
define    parallel term rewriting.

\subsection{Parallel Term Rewriting}

As discussed in the introduction, in the parallel interpretation of
term graph rewriting the reduction of a cyclic graph corresponds to 
the parallel reduction of a possibly infinite set of redexes 
in the corresponding term. 
The definitions below summarize, in a simplified way, those in~\cite{Cor:TRC}.
Intuitively, {\em finite} parallel rewriting can be defined easily by 
exploiting the confluence of orthogonal term rewriting. In fact, 
the parallel reduction of a finite number of redexes is defined simply 
as any 
{\em complete development} of them: Any such development ends with the 
same term, so the result is well defined. Let us recall the relevant 
definitions.
  
Given two redexes of a term, the reduction of one of them can 
transform 
the other in various ways. 
The second redex can be destroyed, it can be left intact, or it can be 
copied a 
number of times. This is captured by the definition of residuals. We 
assume here that the rules belong to an OTRS, thus two redexes in a 
term 
are either the same or do not overlap. 

\begin{definition}[residuals]
\label{de:residuals}
Let $\Delta = (w, R)$ and $\Delta' = (w', R': l' \rightarrow r')$ be 
two redexes in a term $t$. The {\bf set of residual of 
$\Delta$ by $\Delta'$} is denoted by $\Delta \backslash 
\Delta'$, and 
it is defined as

\[
\Delta \backslash \Delta' = \left\{
\begin{array}{lp{8cm}} \emptyset & if $\Delta = \Delta'$\\
\{\Delta\}  & if $w \not > w'$\\
\{(w'w_xu, R) \mid r'/w_x = l'/v_x\}  & if $w = w'v_xu$ and $l'/v_x$
is a variable 
\end{array}
\right.
\]

\noindent
If $\Phi$ is a finite set of redexes of $t$ and $\Delta$ is a redex of 
$t$, then the set of residuals of $\Phi$ by $\Delta$, denoted $\Phi 
\backslash \Delta$, is defined as the union of 
$\Delta' \backslash \Delta$ for all $\Delta' \in \Phi$.

If $\Phi$ is a set of redexes of $t$ and $s = \left(t 
\rightarrow_{\Delta_1} 
t_1 \ldots
\rightarrow_{\Delta_n} t_n\right)$ is a reduction sequence, then $\Phi 
\backslash s$ is defined as $\Phi$ if $n = 0$, and as $(\Phi 
\backslash 
\Delta_1) \backslash s'$, where $s' =  \left(t_1 
\rightarrow_{\Delta_2} t_2 
\ldots
\rightarrow_{\Delta_n} t_n\right)$, otherwise.
\end{definition}

In the last definition, if $t \rightarrow_\Delta t'$ and $\Delta'$ is 
a 
redex of $t$, the orthogonality of the system ensures that every 
member 
of $\Delta' \backslash \Delta$ is a redex of $t'$.

\begin{definition}[complete development]
Let $\Phi$ be a finite set of redexes of $t$. A {\bf development of 
$\Phi$} is 
a 
reduction sequence such that after each initial segment $s$, the next 
reduced redex is an element of $\Phi \backslash s$. A {\bf complete 
development of $\Phi$} is a development $s$ such that $\Phi \backslash 
s 
= \emptyset$.
\end{definition}

The following well-known fact~\cite{DBLP:journals/jacm/BerryL79} is a  
consequence of the {\em parallel moves lemma}~\cite{curry-feys:combinatory-logic}.

\begin{proposition}
All complete developments $s$ and $s'$ of a finite set of redexes 
$\Phi$ in a 
term $t$ 
are finite, and end with the same term. Moreover, for each redex 
$\Delta$ 
of $t$, it holds $\Delta \backslash s = \Delta \backslash s'$. 
Therefore 
we can safely denote by $\Delta \backslash \Phi$ the residuals of 
$\Delta$ by any complete development of $\Phi$ (and similarly 
replacing 
$\Delta$ with a set of redexes $\Phi'$ of $t$). 
\end{proposition}

Exploiting this fact, we define the parallel reduction  
of a finite set of redexes as any complete development of them.

\begin{definition} [finite parallel redex reduction]
\label{de:parallel-rdx}
A {\bf parallel redex} $\Phi$ of a term $t$ is a (possibly infinite, 
necessarily
countable) set of distinct redexes in $t$.  
Given a {\em finite} parallel redex $\Phi$ of $t$, we write $t 
\rightarrow_\Phi t'$ and say that there is a {\bf (finite) parallel 
reduction} 
from $t$ 
to $t'$ if there exists a complete development $t 
\rightarrow_{\Delta_1} 
t_1 \ldots \rightarrow_{\Delta_n} t'$ of $\Phi$.
\end{definition} 

We are now ready to extend the definition of application 
of parallel redexes to the infinite case. 
Given an infinite parallel redex 
$\Phi$ 
(i.e., an infinite set of redexes) of a term $t$, we consider a chain 
of 
approximations of $t$, $t_0 \le t_1 \le t_2 \ldots$, such that their 
limit is $t$, and that only a finite subset of $\Phi$ applies to each 
$t_i$. For each $i < \omega$, let $\Phi_i$ be the finite subset of 
$\Phi$ 
containing all and only those redexes of $t$ which are also redexes of 
$t_i$, and call $d_i$ the result of the parallel reduction of 
$\Phi_i$, 
i.e., $t_i \rightarrow_{\Phi_i} d_i$. Then the crucial fact is that 
the  
sequence of terms $d_0, 
d_1, d_2, \ldots$ defined in this way forms a chain: by definition we 
say 
that there is an infinite parallel reduction from $t$ to $d = 
\bigcup_{i<\omega}d_i$ via $\Phi$, written $t \rightarrow_\Phi d$.
Here is the formal definition.

\begin{definition}
[parallel redex reduction]
\label{de:par-rdx-app}

Given an infinite parallel redex $\Phi$ of a term $t$,
let $t_0 \le t_1 \le \ldots$ $t_n \le \ldots$ be any chain approximating 
$t$ (i.e., such that $\bigcup_{i<\omega}\{t_i\} = t$) and such that:

\begin{itemize}
\item For each $i < \omega$, every redex  $(w, R) \in \Phi$ is either 
a 
redex of $t_i$ or $t_i(w) = \bot$. That is, the image of the lhs of 
every 
redex in $\Phi$ is either all in $t_i$, or it is outside, but does not 
``cross the boundary''.  

\item For each $i < \omega$, let $\Phi_i \subseteq \Phi$ be the subset 
of 
all redexes in $\Phi$ which are also redexes of $t_i$; then $\Phi_i$ 
must be 
finite.
\end{itemize}

\noindent 
For each $i < \omega$, let $d_i$ be the result of 
the (finite) parallel reduction of $t_i$ via $\Phi_i$ (i.e., $t_i 
\rightarrow_{\Phi_i} d_i$). Then we say that there is an 
{\bf (infinite) parallel reduction} from $t$ 
to $d \stackrel{def}{=}  \bigcup_{i<\omega}\{d_i\}$ via $\Phi$, and we 
write 
$t \rightarrow_\Phi d$. 
\end{definition}

Note that in the last definition if the chain approximating $t$
contains finite terms only, then the second condition is automatically
satisfied. We consider more general chains, possibly including
infinite terms, because they arise
naturally in the proof of Theorem \ref{th:soundness}. 
The main result of this section states that the last definition is 
well-given. For a proof we refer to~\cite{Cor:TRC} (Theorem 32).

\begin{theorem}
[parallel redex reduction is well-defined]
\label{th:WD-ParRdxApp}
Definition~\ref{de:par-rdx-app} is well given, that is:

\begin{enumerate}
\item given an infinite parallel redex $\Phi$ of a term $t$, there
  exists a chain $t_0 \le t_1 \le \ldots t_n \le \ldots$ approximating
  $t$ and satisfying the conditions of Definition~\ref{de:par-rdx-app};

\item in this case, $\forall\, 0 \leq i < j < \omega \,.\, d_i \leq
  d_j$, thus $\{d_i\}_{i<\omega}$ is a chain;

\item the result of the infinite parallel reduction of $t$ via $\Phi$
  does not depend on the choice of the chain approximating $t$,
  provided that it satisfies the required conditions.
\end{enumerate}
\end{theorem}

As an example, let us consider again the two infinite reductions 
mentioned 
in the introduction, according to the 
parallel interpretation. 
In the term reduction from $f^\omega$ to $g^\omega$ (corresponding to 
the 
graph reduction of Figure \ref{fi:term-graph-2}), 
there are infinitely many redexes of rule $R_f$ 
in the term $f^\omega$, 
namely at occurrences $\lambda,1,1\cdot 1, \ldots$. Let $\Phi = 
\{\Delta_n 
\stackrel{def}{=}(1^n, R_f)\mid 
n \in \nat\}$ be this infinite parallel redex of $f^\omega$. 
As for the chain of terms approximating 
$t=f^\omega$, let us choose  $t_0 = 
\bot, t_1= f(\bot),\ldots, t_n = f^n(\bot)$. Clearly, for each $n$  
the set $\Phi_n \subseteq \Phi$ of redexes of $t_n$ contains exactly 
$n$ 
redexes. 
For each $n$ we have $t_n = f^n(\bot) \rightarrow_{\Phi_n} g^n(\bot)$, 
and thus, according 
to the definition, the result of the parallel reduction of $f^\omega$ 
via $\Phi$ is 
$\bigcup_{i<\omega}\{g^i(\bot)\} = g^\omega.$ 

In the case of rule $R_I$ and of the circular $I$, choosing a similar 
approximating chain we have   $\Phi =    
 \{\Delta_n 
\stackrel{def}{=}(1^n, R_I)\mid 
n \in \nat\}$, $t_n = I^n(\bot)$ for each $n$, 
$t_n = I^n(\bot) \rightarrow_{\Phi_n} \bot$, and thus $I^\omega$ 
reduces 
by $\Phi$ to $\bigcup_{i<\omega}\{\bot\} = \bot$.  

We shall need the following easy result.
\begin{proposition} 
[strong confluence of parallel reduction]
\label{pr:confluence}
Given an orthogonal TRS $\mathcal{R}$, parallel reduction is strong 
confluent, i.e., if $t'\ _{\Phi'}\!\!\leftarrow t \rightarrow_\Phi t''$, 
then there exist $t''', \Psi, \Psi'$ such that $t' \rightarrow_{\Psi'}
t'''\  _\Psi\!\!\leftarrow t''$. As a consequence, parallel reduction is 
confluent.
\end{proposition}

\section{Term Graphs and Rational Terms}
\label{se:Term Graphs and Rational Terms}

We summarize here the definition of {\em term graphs} (or 
simply 
{\em graphs}), and their relationship with rational terms, as 
introduced 
in~\cite{KKSV:AGRS}. However, since we will apply to those graphs the 
algebraic approach to graph rewriting, we shall slightly adapt the 
definition to our framework, emphasizing the categorical 
structure of the collection of graphs.

Term graphs are obtained from the usual representation 
of terms with sharing as {\em dag's} ({\em directed acyclic graph}), by 
dropping the acyclicity requirement. In such a way, a finite cyclic 
graph 
may represent a possibly infinite, but rational term.
 
\begin{definition}
[term graphs]
\label{de:graphs}
Let $\Sigma$ be a fixed, one-sorted\footnote{The 
generalization to many-sorted signatures is straightforward, by 
labeling nodes with pairs {\em $\langle$operator, sort$\rangle$}, 
with the obvious meaning.} signature. A {\bf (term) graph} $G$
(over $\Sigma$)
  is a triple $G = (N_G,s_G,l_G)$, where
\begin{itemize}
\item	$N_G$ is a {\em finite} set of {\bf nodes\/},
\item	$s_G: N_G \rightarrow N_G^*$ is a partial function, called the 
{\bf 
successor} function,
\item	$l_G:N_G \rightarrow \Sigma$ is a partial function, called the 
{\bf labelling} function.
\end{itemize}
\noindent
Moreover, it is required that $s_G$ and $l_G$ are defined on the same 
subset of $N_G$, and that for each node $n \in N_G$, if $l_G(n)$ is 
defined and it is an operator of arity $k$, then $s_G(n)$ has length 
exactly 
$k$. \end{definition}

\begin{definition}
[morphisms, category of term graphs]
\label{de:morphisms}
A {\bf (graph) morphism} $f: G \rightarrow H$ between two 
graphs
$G$ and $H$ is a function $f: N_G \rightarrow N_H$, which preserves 
labelling and successor functions, i.e., for each $n \in N_G$, if 
$l_G(n)$ is defined, then $l_H(f(n)) = l_G(n)$ and $s_H(f(n)) = 
f^*(s_G(n))$ (where $f^*$ is the obvious extension of $f$ to lists of 
nodes).

The composition of graph morphisms is defined in the obvious way, and 
it 
is clearly associative; moreover, the identity function on nodes is a 
morphism, and therefore term graphs (over $\Sigma$) and their 
morphisms 
form a category that will be denoted by \bf{TGraph$_\Sigma$}.
\end{definition}

Thanks to the conditions imposed on term graphs
in Definition \ref{de:graphs}, a term of $CT_\Sigma$ can be 
{\em extracted} or {\em unraveled} from every node of a graph.

\begin{definition}
[from term graphs to terms and backwards]
\label{de:from term graphs to terms}

A {\bf path} $\pi$ in a graph
$G = (N_G, s_G, l_G)$ from node $n$ to node $n'$ is a finite sequence 
$\pi = \langle n_1, j_1, n_2, j_2, \ldots, j_k,  n_{k+1}\rangle$,
where all $j_i$ are natural numbers, all $n_i$ 
are nodes, and such that $n_1 = n$, $n_{k+1} = n'$, and for all $1 
\leq i 
\le k$, $s_G(n_i)|_{j_i} = 
n_{i+1}$ (here $s|_i$ denotes the $i$-th element of the
sequence $s$). It follows that there exists exactly one {\bf empty 
path} `$\langle n \rangle$' from each node $n$ to itself.
The {\bf occurrence of a path} $\pi =\langle n_1, j_1, n_2, 
j_2, \ldots, j_k,  n_{k+1}\rangle$ is the list of natural numbers 
$j_1\cdot j_2 \cdots j_k$, and it is denoted by $\mathcal{O}(\pi)$. Thus 
$\mathcal{O}(\pi) = \lambda$ iff $\pi$ is an empty path. Clearly, 
for each node $n$ there is at most one path in $G$ having a given 
occurrence $w$. 

A graph $G$ is
{\bf acyclic} if there are  no non-empty paths from one node to 
itself; graph $G$ is a {\bf tree} with root $\underline{n} \in N_G$ 
iff there
exists exactly one path from $\underline{n}$ to any other node of $G$.
  
Let $G = (N_G, s_G, l_G)$ be a term graph. The set $\varnodes{G} 
\subseteq N_G$ 
of {\bf variable
nodes\/} or {\bf empty nodes} of $G$ is the set of nodes on which the 
labeling function (and thus also the successor function) is undefined.
For each node $n\in N_G$, 
$\mathcal{U}_G[n]$, the {\bf unraveling of $G$ at $n$} is the term 
defined as
\[
\mathcal{U}_G[n](w)=
 \left\{
  \begin{array}{l@{\hspace{1cm}}p{10cm}}
n' & if there is a path $\pi$ from $n$ to $n'$ with 
$\mathcal{ O}(\pi) = w$ and $n' \in \varnodes{G}$\\
l_G(n')	& if there is a path $\pi$ from $n$ to $n'$ with 
$\mathcal{ O}(\pi) = w$ and $n' \not \in \varnodes{G}$\\
 \bot & otherwise.
 \end{array}
\right.
\]
for all $w\in\omega^*$. It follows immediately from these definitions 
that for each $n \in N_G$ and for each occurrence $w$, $\mathcal{U}_G[n] 
/ w = \mathcal{U}_G[n']$ iff there exists a path $\pi$ from $n$ to $n'$ 
with $\mathcal{O}(\pi) = w$. By $\mathcal{U}[G]$ we denote the set 
$\mathcal{U}[G]=\{\mathcal{U}_G[n] \mid n\in N_G\}$.

Conversely, let $T$ be a finite set of rational terms, and let us 
denote by $\overline{T}$ its
closure under the subterm relation (i.e., $\overline{T} = \{t \mid t 
\mbox{ is
a subterm of some term in $T$}\}$). Then the 
{\bf term graph representation of $T$}, denoted $\mathcal{G}[T]$, is the 
graph  
$\mathcal{G}[T] = (N_{\mathcal{G}[T]}, s_{\mathcal{G}[T]}, l_{\mathcal{G}[T]})$ 
defined as 
follows:
\begin{enumerate}
\item
$N_{\mathcal{G}[T]} = \overline{T};$
\item 
$s_{\mathcal{G}[T]}(t) = \langle t_1, \ldots, t_k\rangle$ \quad
if $t = f(t_1, \ldots, t_k)$, and undefined if $t$ is a variable;
\item 
$l_{\mathcal{G}[T]}(t) = f$ \quad
if $t = f(t_1, \ldots, t_k)$, and undefined if $t$ is a variable.
\end{enumerate}
\end{definition}

It is quite easily seen that, for every term graph $G$ and 
node $n \in N_G$, $\mathcal{U}_G[n]$ is a well-defined and total term 
containing
variables in $\varnodes{G}$ (thus $\mathcal{U}_G[n] \in 
CT_\Sigma(\varnodes{G}))$.
Such term can be infinite (because in a cyclic
graph there can be infinitely many paths starting from one node), but 
since $G$ is finite by definition, $\mathcal{U}_G[n]$ is
necessarily a rational term, because
the number of distinct subterms is bounded by the 
cardinality of $N_G$. It is worth noting also that if $G$ is a tree 
with
root $\underline{n}$, then $\mathcal{U}_G[\underline{n}]$ is a finite and 
linear 
term.


The above definitions made clear the relationship between objects of
category {\bf TGraph$_\Sigma$} and terms in $CT_\Sigma$. Such 
relationship 
can be extended to arrows of {\bf TGraph$_\Sigma$} and to term 
substitutions 
as follows.

\begin{proposition}
[from morphism to substitutions and backwards]
\label{pr:morphisms}
Let  $f: G \rightarrow
H$ be a term graph morphism. The {\bf substitution induced by $f$} is 
the  
substitution $\sigma_f: \varnodes{G}
\rightarrow CT_\Sigma(\varnodes{H})$, defined as $x\sigma_f = 
\mathcal{U}_H[f(x)]$ for all $x \in \varnodes{G}$. Moreover the following 
hold:
\begin{enumerate}
\item
If $f: G \rightarrow H$ is a morphism, then 
$\sigma_f\circ \mathcal{U}_G = \mathcal{U}_H\circ f$.

\item
Let $G$ be a tree with root $\underline{n}$ and let $\mathcal{
U}_G[\underline{n}]=t$. Then for every
graph $H$ and every node $n'\in N_{H}$ with 
$t'=\mathcal{U}_{H}[n']$,
there is a substitution $\sigma$ such that $t\sigma=t'$ {\em if and 
only if}
there is a morphism $f: G\rightarrow H$ where
$f(\underline{n})=n'$. Moreover in  
this
case we have $\sigma = \sigma_f$.
\end{enumerate}
\end{proposition} 

\section{Algebraic term graph rewriting}
\label{se:Algebraic term graph rewriting}

We introduce now term graph rewriting according to the algebraic, 
double-pushout approach~\cite{Ehr:TIAA}. 
%
Let {\bf Graph} be a fixed category of graphs (below we will apply the 
general definitions to the category of term graphs defined in Section 
\ref{se:Term Graphs and Rational Terms}). The basic categorical 
construction in the 
algebraic definition of graph rewriting is that of pushout.

\begin{definition}
[pushout~\cite{ML:CWM} and pushout complement~\cite{Ehr:TIAA}]
\label{de:pushout} 

Given a category $C$ and two arrows $b:
K \rightarrow B$, $d: K \rightarrow D$ of $C$, a triple  $\langle  H, 
h:
B \rightarrow H, c: D \rightarrow H \rangle $  as in Figure 
\ref{fi:pushout} (a)
is called a {\bf pushout \/} of  $\langle  b, d \rangle $  if 
{\em[Commutativity]} $h  \circ  b = c  \circ  d$, and 
{\em [Universal Property]} for all objects $H'$ and arrows 
$h': B
 \rightarrow  H'$ and $c': D  \rightarrow  H'$, with $h'  \circ  b =
c'  \circ  d$, there exists a unique arrow $f: H  \rightarrow  H' $
such that $f  \circ  h = h'$ and $f  \circ  c = c'$.

In this situation, $H$ is called a {\bf pushout object} of  $\langle  
b, d
\rangle $.
Moreover, given arrows $b: K \rightarrow B$ and $h: B \rightarrow H$,
a {\bf pushout complement \/} of  $\langle  b, h \rangle $  is a 
triple
$\langle  D, d:
K \rightarrow D, c: D \rightarrow H \rangle $  such that  $\langle  H, 
h, c
\rangle $  is
a pushout of $b$ and $d$. In this case $D$ is called a {\bf pushout
complement object \/} of  $\langle  b, h \rangle $.
\end{definition}

\begin{figure}
\begin{center}
$    \xymatrix@=3ex{
     {K} \ar[dd]_d \ar[rr]^b & & {B} \ar[dd]_h \ar[dddr]^{h'} &  \\
 & & &\\
         {D} \ar[rr]^c \ar[drrr]_{c'} & & {H} \ar[dr]|f & \\
         & & & {H'} 
      }
$
\hspace{2cm}
$    \xymatrix@=7ex{
      {(L} \save []-<.5cm, .0cm>*\txt{$p=$} \restore
      \ar[d]_g & {K} \ar@{->}[l]_{l} \ar@{->}[r]^{r} \ar[d]^k & {R)} \ar[d]^h\\
      {G} & {D} \ar[l]^d \ar[r]_b & {H}
      }
$
\end{center}
\caption{{\bf (a)} Pushout diagram\quad {\bf (b)} 
Direct derivation as double-pushout construction}
\label{fi:pushout}
\end{figure}
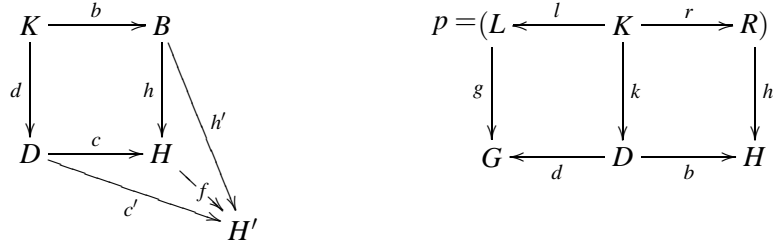

\begin{definition}
[graph grammars, direct derivations~\cite{Ehr:TIAA}]
\label{de:graph-grammar}

A {\bf (graph) production\/} $p = (L \stackrel{l}{\leftarrow} K
\stackrel{r}{\rightarrow} R)$ is a pair of injective 
graph morphisms $l: K \rightarrow L$ and $r: K \rightarrow R$.  The
graphs $L$, $K$, and $R$ are called the {\bf left-hand side},
the {\bf interface}, and the {\bf right-hand side} of $p$,
respectively. 
A {\bf graph transformation system} $\mathcal{G} = \{p_i\}_{i\in
I}$ is a set of graph productions.

Given a graph $G$, a graph production $p = (L
\stackrel{l}{\leftarrow} K \stackrel{r}{\rightarrow}  R)$, and an 
{\bf occurrence\/} (i.e., a graph morphism) $g: L \rightarrow G$, a
{\bf direct derivation $\alpha$ from $G$ to $H$ using $p$ (based on 
$g$)}
exists if and only if the diagram in Figure \ref{fi:pushout} (b)
can be constructed, where both squares are required to be 
pushouts in
{\bf Graph}. In this case, $D$ is called the {\bf context\/} graph,
and we write $\alpha : G \Rightarrow_{p,g} H$, or simply $\alpha : G 
\Rightarrow_p H$.
\end{definition}

In a graph-theoretical setting, the pushout object $H$ of Figure
\ref{fi:pushout} (a) can be understood as the gluing of graphs $B$ and 
$D$, obtained by identifying the images of $K$ along $b$ and $d$. 
Therefore the double-pushout construction can be interpreted as 
follows. In
order to apply the production $p$ to $G$, we first need to find an
occurrence of its left-hand side $L$ in $G$, i.e., a graph morphism 
$g: L
\rightarrow G$. Next, to model the deletion of that occurrence from
$G$, we have to find a graph $D$ and morphisms $k$ and $d$ such that
the resulting square is a pushout: The context graph $D$ is 
characterized categorically as the pushout
complement object of $\langle l, g \rangle$. Finally, we have to
embed the right-hand side $R$ into $D$: This embedding 
is expressed by the right pushout. 

The conditions for the existence of 
pushouts and of pushout complements depend on the category {\bf 
Graphs} 
for which the above definitions are introduced. Since we are 
interested 
just in the graph productions that represent term rewrite rules, in 
the 
rest of the paper we shall consider such definitions in the category 
{\bf TGraph$_\Sigma$}, and we will present conditions for the 
existence
of pushouts and pushout complements (see Proposition \ref{pr:po-poc}) 
only for a specific format of productions, called 
{\em evaluation rules} (according to the name in~\cite{HP:ITRJ}, where 
the
acyclic case is considered). 
Such evaluation rules are term graph productions satisfying 
some additional requirements that make them suitable to represent term 
rules.

\begin{definition}[evaluation rules]
\label{de:evaluation rules}
An {\bf evaluation rule\/} is a term graph production 
$p = (L \stackrel{l}{\leftarrow} K
\stackrel{r}{\rightarrow} R)$ 
such that
\begin{enumerate}

\item  $L$ is a tree and it is not a single empty node. Let 
$\underline{n}$ be
the root of $L$.

\item  $K= (N_L, s_L\downarrow(N_L\backslash \{\underline{n}\}), 
l_L\downarrow(N_L\backslash \{\underline{n}\}))$, that is, $K$ is 
obtained
from $L$ by making the successor and labeling functions undefined on 
the root.
Morphism $l: K \rightarrow L$ is the inclusion; notice that $l$ is an 
isomorphism on nodes and that $\varnodes{K} = \varnodes{L}\cup 
\{\underline{n}\}$.

\item  If restricted to $\varnodes{L}$ ($\subset \varnodes{K}$), morphism $r: 
K \rightarrow R$ is an isomorphism between $\varnodes{L}$ 
and $\varnodes{R}$. 
\end{enumerate}
\end{definition}

Let us make explicit the relationship between evaluation rules and 
rewrite rules. From every evaluation rule 
$p$ one can easily unravel a term rewrite rule $\mathcal{U}[p]$; 
furthermore,
for each rewrite rule $R$ we propose a suitable representation as 
evaluation rule, $\mathcal{G}[R]$.

\begin{definition} 
[from evaluation to rewrite rules and backwards]
\label{de:evaluation to rewrite}

Let $p = (L \stackrel{l}{\leftarrow} K
\stackrel{r}{\rightarrow} R)$ be an evaluation rule. The {\bf 
unraveling} of $p$ is the term rewrite rule $\mathcal{U}[p]: t 
\rightarrow s$ defined as 
follows.
\begin{enumerate}

\item
$t = \mathcal{U}_L[\underline{n}]$ (where $\underline{n}$ is the 
root of $L$, as usual);

\item
$s = \mathcal{U}_R[r(\underline{n})]\sigma$, where $\sigma: 
\varnodes{R} \rightarrow \varnodes{L}$ is the substitution defined as
$  \sigma(x) = y \mbox{ \ \ if } r(y) = x$.
\end{enumerate}

\noindent 
We shall say that an evaluation rule $p$ is {\bf non-self-overlapping}
if so is the rewrite rule $\mathcal{U}[p]$. A {\bf term graph rewriting 
system} (shortly {\bf TGRS}) $\mathcal{P}$ is a finite set of evaluation 
rules, $\mathcal{P} = \{p_i\}_{i\in I}$; $\mathcal{P}$ is called {\bf 
orthogonal} if so is the term rewriting system $\mathcal{U}[\mathcal{P}] 
\stackrel{def}{=} \{\mathcal{U}[p] \mid p \in \mathcal{P}\}$. 

The other way around, let $R: t \rightarrow s$ be a rational,
left-finite and left-linear term rewrite rule, and let $t = f(t_1,
\ldots, t_k)$ ($t$ must have this form because it cannot be a
variable). Then its {\bf graph representation} $\mathcal{G}[R]$ is the
production $\mathcal{G}[R] = (L \stackrel{l}{\leftarrow} K
\stackrel{r}{\rightarrow} R)$, where

\begin{enumerate}
\item $L = \mathcal{G}[\overline{\{t\}}]$ (by $\overline{T}$ we denote 
the closure of $T$ with respect to the subterm relation).

\item $K$ and $l : K \rightarrow L$ are defined according to point 2 
of 
Definition \ref{de:evaluation rules}.

\item $R = \mathcal{G}[\overline{\{s, t_1, \ldots, t_k\}}]$.

\item $r: K \rightarrow R$ is defined as $r(\underline{n}) = s$, and 
$r(n) = \mathcal{U}_K[n]$ if $n \in N_K\backslash\{\underline{n}\}$ (this 
is 
well-defined because the nodes of $R$ are subterms of $\{s, t_1, 
\ldots,
 t_k\}$, according to Definition \ref{de:from term graphs to 
terms}).
\end{enumerate}
\end{definition}

By the properties of the unraveling function and of evaluation rules 
(Definitions 
\ref{de:from term graphs to terms} and \ref{de:evaluation rules})
it is routine to check that for each 
evaluation rule $p$ the term rewrite rule $\mathcal{U}[p]$ is 
well-defined,
rational, total, left-finite, and left-linear (the last two because 
the left-hand side of an evaluation rule is a tree). The substitution 
$\sigma$ applied to $\mathcal{U}_R[r(\underline{n})]$ in the above 
definition is needed to ensure that $var(s) \subseteq var(t)$.
Similarly, it follows directly from the definitions that $\mathcal{G}[R]$ 
is 
a well-defined evaluation rule for each term rewrite rule $R$ 
satisfying the 
required conditions.

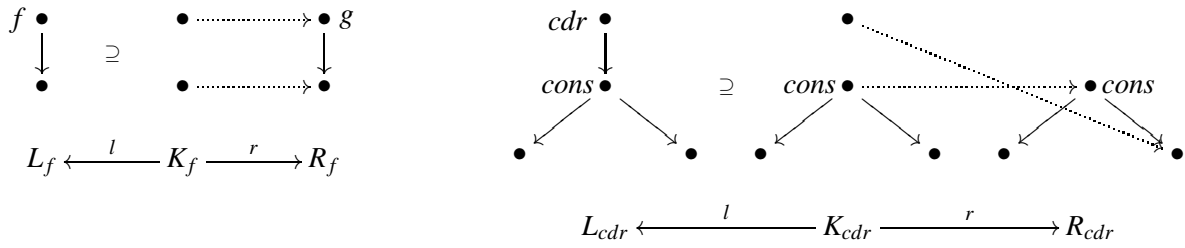
\begin{figure}[t]
$$    \SelectTips {cm}{}\xymatrix@=3ex{
{ \save []-<.3cm, .0cm>*\txt{\emph{f}}\restore \bullet \ar[d]}
     &     &  \bullet\ar@{}[dll]|\supseteq \ar@{..>}[rr]
                     &    &   {\save []+<.3cm, .0cm>*\txt{\emph{g}} \restore \bullet \ar[d]}   
                                \\
  \bullet
     &     & \bullet\ar@{..>}[rr]
                     &    &  \bullet\\
L_f & & K_f \ar[ll]_l \ar[rr]^r & & R_f 
}
\hspace{2cm}    \SelectTips {cm}{}\xymatrix@=3ex{
    &  { \save []-<.5cm, .0cm>*\txt{\emph{cdr}}\restore \bullet \ar[d]}
          &     &     &  {\bullet\ar@{..>}[rrrrdd]}  &     &     &    \\
    &  { \save []-<.5cm, .0cm>*\txt{\emph{cons}}\restore \bullet \ar[dl] \ar[dr]}
          &     &     &   { \save []-<.5cm, .0cm>*\txt{\emph{cons}}\restore \bullet \ar[dl] \ar[dr] \ar@{..>}[rrr]\ar@{}[lll]|\supseteq}
                           &     &   & { \save []+<.5cm, .0cm>*\txt{\emph{cons}}\restore \bullet \ar[dl] \ar[dr]}
                                       &    \\
{ \bullet}   
    &     &  {\bullet}
                &   {\bullet}
                      &    &  {\bullet}
                                 &  {\bullet}
                                       &    & {\bullet}\\
&L_{cdr} & & & K_{cdr} \ar[lll]_{l} \ar[rrr]^r & & & R_{cdr} 
}
$$

  \caption{The evaluation rules ${\cal
  G}[R_f]$ and $\mathcal{G}[R_{cdr}]$}
  \protect\label{fi:evaluationRules}
\end{figure}

\begin{example}[evaluation rules]
\label{ex:rules}
Figure~\ref{fi:evaluationRules} shows the evaluation rules ${\cal
  G}[R_f]$ and $\mathcal{G}[R_{cdr}]$, which are the graph
representations of the rewrite rules $R_f: f(x) \to g(x)$ and
$R_{cdr}: cdr(cons(x,y)) \to y$. The left morphisms of the rules are
the obvious inclusions, while the right morphisms are determined by
the mapping of nodes that is depicted with dotted arrows.
Rule $R_{cdr}$ is a collapsing rule which
describes the behaviour of the $cdr$ operator on {\sc LISP}-like lists
built with the pairing operator $cons$.
\end{example}

The next proposition ensures that if we consider a 
non-self-overlapping  evaluation rule, then 
the existence of an occurrence morphism from its 
left-hand side to a term graph is a sufficient condition for the 
applicability of the rule, i.e., the pushout complement and the 
pushout 
of Figure \ref{fi:pushout} (b) can always be constructed.

\begin{proposition}
[existence of pushout complements and pushouts]
\label{pr:po-poc}
Let $p = (L \stackrel{l}{\leftarrow} K
\stackrel{r}{\rightarrow} R)$ be an evaluation which is not 
self-overlapping,\footnote{Without this condition the statement would 
not be true, because the {\em identification condition}~\cite{Ehr:TIAA} 
may not be satisfied.} and let $g: L 
\rightarrow G$ be an occurrence morphism. Then in category {\bf 
TGraph}
there exists a pushout complement $\langle  D, k:
K \rightarrow D, d: D \rightarrow G \rangle $
 of $\langle l, g\rangle$, where $D$ is 
obtained from graph $G$ by making the labeling and the successor 
function 
undefined on the image of the root of $L$. 

Moreover, a pushout $\langle H, h: R \rightarrow H, b: D \rightarrow H 
\rangle$ of the resulting arrow $k: K \rightarrow D$ and $r$ 
always exists, and therefore in the above hypotheses there exists a 
direct derivation $G \Rightarrow_{p,g} H$.
\end{proposition}

The proof of the last proposition is reported (for the equivalent 
category  {\bf Jungle}) in~\cite{CR:HRJR}, where also general conditions 
for the existence of pushouts are presented.\footnote{An interesting 
fact, which is not relevant for this paper, is that the pushout of two 
arrows exists in {\bf TGraph$_\Sigma$} iff the associated 
substitutions 
unify, and in this case the pushout is a most general unifier.} 
It is worth stressing that in the hypotheses of the last proposition,
the pushout 
complement object $D$ has the same nodes of $G$; thus the nodes of $G$ 
can be ``traced'' after the rewriting. More formally, there is a total 
function, called the {\bf track function\/}~\cite{HP:ITRJ} $tr: N_G 
\rightarrow N_H$,  
defined as $tr(n)=b(n)$ for all $n\in N_G$ (= $N_D$).

\section{Adequacy of algebraic term graph rewriting for rational 
parallel term rewriting}
\label{se:adequacy}

The relationship between term and term graph rewriting has been nicely 
formalized in~\cite{KKSV:AGRS} with the notion of {\em adequate mapping} 
between rewriting systems. We recall here the definition, referring to 
that paper for the precise motivations. A {\em rewriting system} is 
defined in this context as a triple $(A, R, S)$, where $A$ is a set of 
states (in our case terms or graphs), $R$ is a set of rules, and $S$ 
is a set of reduction sequences, closed under certain operations.   

\begin{definition} 
[adequate mapping between rewriting systems]
\label{de:adequate mapping}
Let $(A_1,$ $ R_1,$ $ S_1)$ and $(A_2, R_2, S_2)$ be two rewriting systems. 
A mapping $\mathcal{U}: A_1 \rightarrow A_2$ is {\bf adequate} if:

\begin{description}
\item {[{\em Surjectivity}]} $\mathcal{U}$ is surjective;
\item {[{\em Preservation of normal forms}]} $a \in A_1$ is in normal 
form iff $\mathcal{U}[a] \in A_2$ is in normal form.\footnote{A state is 
in 
normal form if there is no reduction sequence starting from it.}

\item {[{\em Preservation of reductions}]} If $a \rightarrow^{*} a'$ 
with a reduction sequence in $S_1$, then $\mathcal{U}[a] \rightarrow^{*} 
\mathcal{U}[a']$ with a reduction sequence in $S_2$.\footnote{We consider 
only reduction sequences of finite length.} 

\item {[{\em Cofinality}]} For $a \in A_1$ and $b \in A_2$, if $\mathcal{
U}[a] \rightarrow^* b $ in $S_2$, then there is an $a'$ in $A_1$ such 
that $a \rightarrow^* a'$ in $S_1$ and $b \rightarrow^* \mathcal{U}[a']$ 
in $S_2$.
\end{description}
\end{definition}

We show below that the unraveling function $\mathcal{U}$ introduced in 
the previous sections is an adequate mapping from a given orthogonal 
TGRS 
$\mathcal{P}$ to its unraveled orthogonal TRS $\mathcal{U}[\mathcal{P}]$, 
restricting the allowed parallel reduction sequences to the rational 
ones. Intuitively, the restriction to \emph{rational} parallel reductions is
justified by the fact that an occurrence of an evaluation rule $p$ in a
term graph $G$ induces a possibly infinite, but certainly
rational parallel redex in the term obtained by unraveling $G$.

\begin{definition} 
[rational parallel reduction sequences]
\label{de:rational sequences}
Let $\mathcal{R}$ be an orthogonal TGR. A parallel reduction $t 
\rightarrow_\Phi t'$ is {\bf rational} if the term obtained by marking 
in $t$ all the occurrences of redexes of $\Phi$ is rational. A {\bf 
rational parallel reduction sequence} is a   parallel reduction 
sequence where each step is rational. A {\bf rational} TRS is a triple
$(CT_\Sigma^{rat}, \mathcal{R}, \mathcal{S}_{rat}(\mathcal{R}))$, where 
$CT_\Sigma^{rat}$ is 
the set of all rational terms in $CT_\Sigma$, $\mathcal{R}$ is an 
orthogonal TRS where all right-hand sides are rational terms, and 
$\mathcal{S}_{rat}(\mathcal{R})$ is the set of all rational parallel 
reduction sequences 
using rules in $\mathcal{R}$.
\end{definition}

It is worth stressing that strong confluence and confluence 
(Proposition \ref{pr:confluence}) also hold for rational parallel 
reductions. 

In order to prove the adequacy of the unraveling function, we need 
the following fundamental result, which formalizes the ``parallel
interpretation'' discussed in the introduction: 
A single (possibly cyclic) term graph 
reduction can be interpreted as a parallel (possibly infinite) term 
reduction.

\begin{theorem} 
[from graph reductions to parallel term reductions]
\label{th:soundness}
Let $p = 
(L \stackrel{l}{\leftarrow} K \stackrel{r}{\rightarrow} R)$ be a 
non-self-overlapping evaluation rule, let $G$ be a term graph, and let 
$g: L \rightarrow 
G$ be an occurrence morphism. By Proposition \ref{pr:po-poc}
we know that $G \Rightarrow_{p, g} H$, with the corresponding track 
function $tr: N_G \rightarrow N_H$. Then for each node $n \in N_G$, we 
have 
$$\mathcal{U}_G[n] \rightarrow_\Phi \mathcal{U}_H[tr(n)]\sigma$$
\noindent
where $\Phi$ is the (possibly infinite) parallel redex $\Phi = 
\{(\mathcal{O}(\pi), \mathcal{U}[p])  \mid \pi$ is a path in $G$ 
from $n$ to  $g(\underline{n}) \}$ ($\underline{n}$ is the root of 
$L$), and where substitution $\sigma: 
var(H) \rightarrow var(G)$ is defined as $\sigma(x) = y$ if $y \in 
var(G) \wedge tr(y) = x$, and $\sigma(x) = \bot$ if $\not \exists y\in 
\varnodes{G} \,.\, tr(y) = x$.
\end{theorem}

\begin{proof} 
For the sake of simplicity, let us assume that $var(G) = \emptyset$, 
which implies that substitution $\sigma$ becomes $x\sigma = \bot $ for 
all $x \in var(H)$.\footnote{The general case needs an 
additional technical lemma showing that the trace function maps 
variables to variables in an injective way (thus $\sigma$ is 
well-defined); this can be proved by a careful inspection of the 
double-pushout diagram.}

Let $R: l_p \rightarrow r_p$ be the rewrite rule $\mathcal{
U}[p]$. We first have to show that $\Phi$ is a parallel redex of 
$\mathcal{U}_G[n]$, i.e., that for each $(\mathcal{O}(\pi), R) \in \Phi$, 
there is a substitution $\tau$ such that $\mathcal{U}_G[n] / \mathcal{
O}(\pi) = l_p\tau$. In fact we have

\[\begin{array}{llp{10cm}}
\mathcal{U}_G[n] / \mathcal{O}(\pi) & = & [because $\pi$ is a path from $n$ 
to $g(\underline{n})$ (see Definition \ref{de:from term graphs to 
terms})]\\
\mathcal{U}_G[g(\underline{n})] & = & [by point 1 of Proposition 
\ref{pr:morphisms} and by Definition \ref{de:evaluation to rewrite}]\\
\mathcal{U}_L[\underline{n}]\sigma_g = l_p\sigma_g& &
\end{array}\]

\noindent
Thus all redexes in $\Phi$ are realized by the same substitution 
$\sigma_g$. Let $W$ be the set of all occurrences of paths from $n$ to 
$g(\underline{n})$ in $G$ (thus $W$ is the set of all occurrences of 
redexes in $\Phi$), and let $\langle w_1, w_2, \ldots \rangle$ be an 
arbitrary but fixed enumeration of $W$ such that if $w_i < w_j$, then 
$i < j$ (in words, no occurrence can be followed by one of its prefixes). For all 
$i < \omega$, define

\[t_i(u) = \left\{\begin{array}{l@{\hspace{1cm}}p{7cm}}
\mathcal{U}_G[n](u) & if $\forall j>i\,.\,u\not\ge w_j$\\
\bot & \mbox{otherwise.}
\end{array}
\right. \]

\noindent
Obviously, $\{t_i\}_{i< \omega}$ is a chain and $\bigcup_{i< \omega} 
\{t_i\} = \mathcal{U}_G[n]$. Furthermore, chain  $\{t_i\}_{i< \omega}$
satisfies the conditions of Definition \ref{de:par-rdx-app}: by orthogonality,
$(w_j, R)$ is a redex of $t_i$ if $j \leq i$,  while $t_i(w_j) = \bot$ if
$j > i$; and 
the subset $\Phi_i$ of $\Phi$
including all redexes of $t_i$ is finite (more precisely, $\Phi_i =
\{(w_1, R), \ldots, (w_i, R)\}$, and in particular $\Phi_0 =
\emptyset$). Thus by Definition \ref{de:par-rdx-app} we have ${\cal
  U}_G[n] \rightarrow_\Phi \bigcup_{i< \omega} \{s_i\}$, where, for
each $i$, $t_i \rightarrow_{\Phi_i} s_i$.

It remains to prove that $\bigcup_{i< \omega} \{s_i\} = 
\mathcal{U}_H[tr(n)]\sigma$. Let us first show by induction that $\forall 
i < \omega \,.\, s_i \leq \mathcal{U}_H[tr(n)]\sigma$.\footnote{It is 
worth recalling that $t \leq t' \Leftrightarrow \forall u \in \mathcal{
O}(t)\,.\, t(u) = t'(u)$.} 

\begin{description} 

\item {[Base Case]}
Since $\Phi_0 = \emptyset$, we have that $s_0 = t_0$. Let $u \in \mathcal{
O}(s_0)$. By definition $u \not \geq w$ for all $w \in W$. Thus we 
have

\noindent
$\begin{array}[b]{rcp{7.8cm}}
s_0(u) = \mathcal{U}_G[n](u) &=& [assuming that $\pi$ is the only path 
from $n$ to $n'$ in $G$ with occurrence $u$, and since $var(G) = 
\emptyset$]\\
l_G(n') 		&=& [by the explicit definition of $D$ (Proposition 
\ref{pr:po-poc}), since $n' \not = g(\underline{n})$]    \\
l_D(n') 		&=& [by properties of morphisms and definition of 
$tr$]\\
l_H(b(n')) = l_H(tr(n')) &=&  [because the track function preserves 
all paths (like $\pi$) not containing $g(\underline{n})$]\\
\mathcal{U}_H[tr(n)](u)		&=& [$\sigma$ does not affect occurrences 
of operators, like $u$]\\ 
\mathcal{U}_H[tr(n)]\sigma(u)
\end{array}$
  
\item {[Inductive Case]} We must show that $s_i \leq \mathcal{
U}_H[tr(n)]\sigma \Rightarrow s_{i+1} \leq \mathcal{U}_H[tr(n)]\sigma$. 
By the above definitions, the only redex of $t_{i+1}$ which is not of 
$t_i$ is $\Delta_{i+1} = (w_{i+1}, R)$. If $t_{i+1} 
\rightarrow_{\Phi_i} t'$, let $\Delta_{i+1} \backslash \Phi_i$ be the 
residual, and let $V = \{v_1, \ldots, v_k\}$ its set of occurrences. 
Then, if $\tau$ is the substitution such that $l_p\tau = t_{i+1}/ v$ 
for all $v \in V$,\footnote{It can be checked that all redexes in 
$\Delta_{i+1} \backslash \Phi_i$ are realized by the same substitution 
in $t'$.}  a careful inspection reveals that $s_{i+i}$ can be defined 
in term of $s_i$ as

$s_{i+1}(u) = \left\{\begin{array}{l@{\hspace{1cm}}p{10cm}}
s_i(u) & if $u \in \mathcal{O}(s_i)$\\
r_p\tau(w) & if there is a $v \in V$ such that $u = vw$, and $w \in 
\mathcal{O}(r_p\tau)$\\
\bot & otherwise.
\end{array}
\right. $

\noindent
Exploiting the induction hypothesis, it remains to show that 
$s_{i+1}(vw) = \mathcal{U}_H[tr(n)]\sigma(vw)$ only for $v \in V$ and $w 
\in \mathcal{O}(r_p\tau)$. This can be done by combining the techniques 
used for the acyclic case in~\cite{HP:ITRJ,CR:HRJR}, and those used for the 
Base Case above, because the definition of $t_{i+1}$ ensures 
that no redex of $\Phi$ appears in substitution $\tau$ (although 
$\tau$ 
may substitute an infinite term for a variable).

\end{description}

\noindent 
Finally it remains to prove that $\mathcal{U}_H[tr(n)]\sigma \leq 
\bigcup_{i< \omega} \{s_i\}$, and this can be done as follows. If rule 
$R$ is not collapsing, it can be shown that $\bigcup_{i< \omega} 
\{s_i\}$ is a total term (and therefore it is a maximal element of the
approximation ordering), because the least depth of $\bot$'s in $s_i$ 
tends to infinity. If instead $R$ is collapsing, then it is easy to 
check that $s_0 = s_1 = \ldots = \bigcup_{i< \omega} \{s_i\}$, and 
that $\mathcal{U}_H[tr(n)]\sigma = s_0$ using the Base Case above and the 
fact that the only variable in $H$, $tr(g(\underline{n}))$, is 
substituted for $\bot$ by $\sigma$. 
\end{proof}

Exploiting this main result, we can prove the adequacy of the 
unraveling function. One minor problem is due to the fact that the 
unraveling of a term graph is not a term, but a set of terms. Thus 
either we consider TRS where a state can be a set of terms (which 
seems quite unnatural), or we have to consider {\em pointed} term 
graphs, i.e., graphs with a distinguished node, and to assume that the 
$\mathcal{U}$ unravels each graph at that specific node. However, in this 
last case, unraveling would not preserve normal forms, because a graph 
may contain a redex in a part that is not reachable from the 
distinguished node (in the {\em garbage}), and such redex would not 
appear in the unraveled term. Full adequacy could be recovered by 
adding garbage collection at each graph rewriting step (as it is done 
in~\cite{BEGKPS:TGR,KKSV:AGRS}) but we prefer to prove a 
slightly weaker result without modifying the graph rewriting formalism.

\begin{theorem} 
[adequacy of TGR for rational parallel TR]
\label{th:adequacy}
Given an orthogonal TGRS $\mathcal{P}$, let $(\mathcal{G}, \mathcal{P}, \mathcal{
S}(\mathcal{P}))$ be the rewriting system where $\mathcal{G}$ is the set of 
all pointed term graphs over $\Sigma$ and $\mathcal{S}(\mathcal{P})$ is the 
set of all graph reduction sequences using rules in $\mathcal{P}$ and 
such that the track function preserves the distinguished node. 

Then the unraveling function $\mathcal{U}$ is a mapping from  $(\mathcal{G}, 
\mathcal{P}, \mathcal{S}(\mathcal{P}))$ to the rational TRS $(CT_\Sigma^{rat}, 
\mathcal{
U}[\mathcal{P}],$ $\mathcal{S}_{rat}(\mathcal{U}[\mathcal{P}]))$ which satisfies all 
the conditions of adequate mappings (Definition \ref{de:adequate 
mapping}), except of the preservation of normal forms. Nevertheless, 
it satisfies the following weaker version: {\em [Weak preservation of 
normal forms]} if $a \in A_1$ is a normal form, then so is $\mathcal{
U}[a] \in A_2$.    

\end{theorem}

\begin{proof}
According to Definition \ref{de:adequate mapping}, we have:
\begin{description}
\item {[{\em Surjectivity}]} Immediate by Definition \ref{de:from term 
graphs to terms}, using pointed graphs.

\item {[{\em Weak preservation of normal forms}]} Follows  
from point 2 of Proposition \ref{pr:morphisms}: if there is a redex of 
rule $\mathcal{U}[p]$ in term $\mathcal{U}[G]$, then there is an occurrence 
morphism from the left-hand side of $p$ to $G$.  
\item {[{\em Preservation of reductions}]} By repeated applications of 
Theorem \ref{th:soundness}.

\item {[{\em Cofinality}]} Let $G$ be a pointed graph, and suppose 
that $\mathcal{U}[G] \rightarrow^* t$ via a rational reduction sequence. 
We show by induction on the length of the sequence that there is a 
$G'$ such that $t \rightarrow^* \mathcal{U}[G']$, and $G \Rightarrow^* 
G'$.

\begin{description} 

\item {[Base Case]} If  $\mathcal{U}[G] \rightarrow_\Phi t$ and $\Phi$ is 
rational, it is possible to show that there exists a finite set $M$ of 
occurrence morphisms of rules of $\mathcal{P}$ in $G$, such that their 
``unraveling'' (that we do not define formally) $\mathcal{U}[M]$ is a 
parallel redex of $\mathcal{U}[G]$ containing $\Phi$. Then let $G'$ be 
the graph obtained 
by applying to $G$ all the occurrence morphisms in $M$ in any order 
(by orthogonality $G'$ is well-defined). Then it can be shown that $t 
\rightarrow_{\mathcal{U}[M] \backslash \Phi} \mathcal{U}[G']$.

$$\xymatrix@C=6ex@R=4ex{
\mathcal{U}[G] \ar[r]^>n & t' \ar[r]^>{*} _>{d} \ar[d]_>>{\Phi}
& \mathcal{U}[G'] \ar[d]^>>{\Phi \backslash d} & \\
& t \ar[r] ^>{*} _>>>>{d\backslash \Phi} & t'' \ar[r]^>{*} & \mathcal{U}[G'']\\
G \ar@{=>}[rr]^>{*} & &  G' \ar@{=>}[r] ^>{*} & G''
}$$ 

\item {[Inductive Case]}  Suppose that $\mathcal{U}[G] \rightarrow^{n} t' 
\rightarrow_\Phi t$, and consider the diagram above.
By 
inductive hypothesis there exists $G'$ such 
that $G \Rightarrow^* G'$ and $t' \rightarrow^* \mathcal{U}[G']$. By 
strong confluence of rational parallel reductions, there exists a 
$t''$ such that the square commutes, and the residual $\Phi 
\backslash d$ is a rational parallel redex. Then applying the Base 
Case to $\mathcal{U}[G'] \rightarrow_{\Phi \backslash d} t''$, we have 
that there exists a $G''$ such that $G' \Rightarrow^* G''$ and $t'' 
\rightarrow^* \mathcal{U}[G'']$. 
\end{description}
\end{description}
\end{proof}

\begin{example}[Rewriting steps]
The diagram below
shows a few direct derivations using the evaluation rule $\mathcal{G}[R_{cdr}]$ of Example~\ref{ex:rules}. The corresponding track functions are uniquely determined by the dotted arrows and the fact that the $cons$ node is preserved. Unraveling the four graphs at the nodes $\circ$, we get the following terms: $t_0= \mathcal{U}_{G_0}[\circ] = cdr(cons(f(a),cdr(cons(f(a),cdr(\ldots)))))$, $t_1 = \mathcal{U}_{G_1}[\circ] = t_0$, $t_2= \mathcal{U}_{G_2}[\circ] = cdr(cons(f(a),\bot))$, and $t_3= \mathcal{U}_{G_3}[\circ] = \bot$. By Theorem~\ref{th:soundness}, each direct derivation $G_i \stackrel{[k]}{\Rightarrow} G_j$  corresponds to a rational parallel reduction $\mathcal{U}_{G_i}[\circ] \to_{\Phi_k} \mathcal{U}_{G_j}[\circ]$: it is easy to check that the four rational parallel redexes are $\Phi_1 = \{\lambda\}$, $\Phi_2 = \{12(12)^*\}$, $\Phi_3 = \{(12)^*\}$ and $\Phi_4 = \{\lambda\}$. 
To conclude, note that if we consider also the evaluation rule $\mathcal{G}[R_f]$ of Example~\ref{ex:rules}, then term graph $G_3$ is not a normal form, while $t_3 =  \mathcal{U}_{G_3}[\circ] = \bot$ is, showing that unraveling does not reflect normal forms in general.
\end{example}
$$\SelectTips {cm}{}
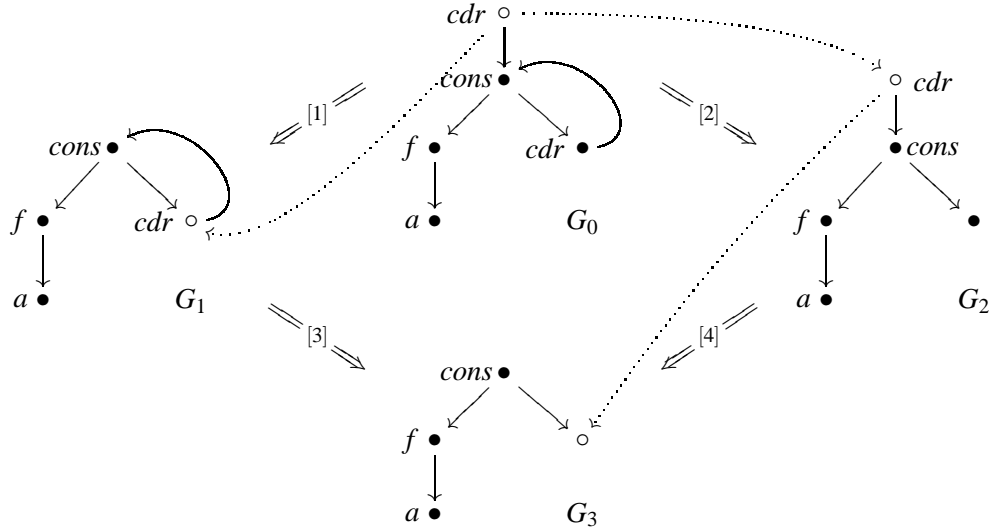
\xymatrix@=3ex{
    &    &    &    &    &    &    &   \save []-<.5cm, .0cm>*\txt{\emph{cdr}} \restore  \circ \ar[d] { \ar@(dl,dr)@{..>}[dddlllll]  \ar@(r,ul)@{..>}[drrrrrr]}
                                        &    &    &    &    &    &    &    \\
    &    &    &    &    & \mbox{}\ar@{=>}[dll]|{\ [1]\ }
                             &    &\save []-<.5cm, .0cm>*\txt{\emph{cons}} \restore \bullet \ar[dl] \ar[dr]
                                       &    &  \mbox{}\ar@{=>}[drr]|{\ [2]\ }   
                                                 &    &    &    &  \save []+<.5cm, .0cm>*\txt{\emph{cdr}} \restore \circ \ar[d]  \ar@/_/@{..>}[dddddlllll]&    \\
    & \save []-<.5cm, .0cm>*\txt{\emph{cons}} \restore \bullet \ar[dl] \ar[dr]
         &    &\mbox{}&    &    &\save []-<.3cm, .0cm>*\txt{\emph{f}} \restore \bullet \ar[d]
                                  &    & \save []-<.5cm, .0cm>*\txt{\emph{cdr}} \restore \bullet \ar@(r,ur)[ul]
                                            &    &    &    &    &  \save []+<.5cm, .0cm>*\txt{\emph{cons}} \restore \bullet \ar[dl] \ar[dr]
                                                                     &     \\
\save []-<.3cm, .0cm>*\txt{\emph{f}} \restore \bullet \ar[d]    
    &    & \save []-<.5cm, .0cm>*\txt{\emph{cdr}} \restore \circ  \ar@(r,ur)[ul]
              &    &    &    & \save []-<.3cm, .0cm>*{\txt{\emph{a}}} \restore \bullet
                                  &    &   G_0 &    &    &    &  \save []-<.3cm, .0cm>*\txt{\emph{f}} \restore \bullet \ar[d] 
                                                                &    &  \bullet 
                                                                          \\
\save []-<.3cm, .0cm>*\txt{\emph{a}} \restore \bullet
    &    & G_1  &  \mbox{} \ar@{=>}[drr]|{\ [3]\ }
                   &    &    &    &    &    &    &    & \mbox{} \ar@{=>}[dll]|{\ [4]\ }
                                                           &  \save []-<.3cm, .0cm>*\txt{\emph{a}} \restore \bullet
                                                                &      & G_2 \\
    &    &    &    &    & \mbox{}
                             &    &  \save []-<.5cm, .0cm>*\txt{\emph{cons}} \restore \bullet \ar[dl] \ar[dr]
                                       &    & \mbox{}
                                                 &    &    &    &    &     \\
    &    &    &    &    &    &\save []-<.3cm, .0cm>*\txt{\emph{f}} \restore \bullet \ar[d]
                                  &    & {\circ} &    &    &    &    &    &   \\
    &    &    &    &    &    & \save []-<.3cm, .0cm>*\txt{\emph{a}} \restore \bullet   &    &  G_3   &    &    &    &    &    &    
}
$$

\section{Conclusions}
\label{se:conclusion}

We showed that by exploiting the complete partial ordered structure of
(infinite, partial) terms, a notion of infinite parallel reduction can
be defined. Moreover, we proved that this notion can be used to relate
term graph rewriting with cycles and term rewriting, formalizing the
intuition that a single graph reduction correspond to a (possibly
infinite) parallel term reduction. This result is used to
show that cyclic graph rewriting is adequate for rational term
rewriting, exploiting the notion of adequacy proposed in~\cite{KKSV:AGRS}. As discussed in the introduction, for several
approaches to cyclic term graph rewriting it is pretty clear whether
this parallel interpretation is more faithful than the sequential one:
this point has still to be clarified for the more recent approaches, 
and it will be addressed in the full version of the paper. Another
interesting topic to explore is how far and under which additional
restrictions the proposed results could be
generalized to non-orthogonal systems. 


\bibliographystyle{eptcs}
\bibliography{TGR}

\end{document}